\documentclass[10pt,a4paper]{amsart}
	
\usepackage{amsmath}
\usepackage{amsfonts}
\usepackage{amssymb}
\usepackage{amsthm}
\usepackage{url}

\numberwithin{equation}{section}

\newtheorem{thm}{Theorem}[section]
\newtheorem{pro}[thm]{Proposition}
\newtheorem{lem}[thm]{Lemma}
\newtheorem{cor}[thm]{Corollary}

\newtheorem{prb}[thm]{Problem}
\theoremstyle{definition}
\newtheorem{dfn}[thm]{Definition}
\newtheorem{rem}[thm]{Remark}
\newtheorem{exem}[thm]{Example}
\theoremstyle{plain}

\DeclareMathOperator{\supp}{supp}
\DeclareMathOperator{\charac}{char}
\DeclareMathOperator{\ord}{ord}

\DeclareMathOperator{\wt}{wt}
\DeclareMathOperator{\Z}{Z}

\newcommand{\Zz}{\mathbb{Z}}
\newcommand{\Cc}{\mathbb{C}}

\begin{document}
\title{Good cyclic codes and the uncertainty principle}
\author{Shai Evra,  Emmanuel Kowalski,  Alexander Lubotzky}
\maketitle

\begin{abstract}
  A long standing problem in the area of error correcting codes asks
  whether there exist good cyclic codes.  Most of the known results
  point in the direction of a negative answer.

  The uncertainty principle is a classical result of harmonic analysis
  asserting that given a non-zero function $f$ on some abelian group,
  either $f$ or its Fourier transform $\hat{f}$ has large support.

  In this note, we observe a connection between these two subjects.
  We point out that even a weak version of the uncertainty principle
  for fields of positive characteristic would imply that good cyclic
  codes do exist.  We also provide some heuristic arguments supporting
  that this is indeed the case.
\end{abstract}

\section{Introduction} \label{s-1}

Let $F$ be a field. Given integers $n$, $k$ and $d$ with
$1\leq k\leq n$, an $[n,k,d]_F$-code, or code over $F$, is a subspace
$C$ of $F^n$ of dimension $\dim_F(C)=k$, such that for every
$0 \ne \alpha \in C$, we have $\wt(\alpha)\geq d$, where the
\emph{weight} $\wt(\alpha)$ of a vector
$\alpha=(a_0,\ldots,a_{n-1})\in F^n$ is the number of non-zero
components $a_i$.  The integer $d$ is called the \emph{distance} of
the code $C$.
\par
Furthermore, a code $C$ is called \textit{cyclic} if it is invariant
under cyclic permutations of the coordinates, i.e.  if
$$
(a_0,\ldots,a_{n-1})\in C \Leftrightarrow
(a_{n-1},a_0,\ldots,a_{n-2})\in C
$$ 
(see~\cite[Ch. 8]{roth}).

The code $C$, or more properly a family $(C_n)$ of codes in $F^n$
where $n\rightarrow\infty$, possibly along some subsequence of
positive integers, is called \textit{good} if there exists a constant
$c > 0$ such that
\begin{equation}\label{eq-good-code}
\frac{k}{n}\geq c,\quad\quad \frac{d}{n}\geq c
\end{equation}
for all $n$.

We are interested in the case of cyclic codes over a finite field $F$
with $\ell$ elements.  The practical interest of such codes goes back
at least to Brown and Peterson~\cite{BP} (e.g., they can be used to
efficiently detect so-called ``burst errors''). A long standing open
problem in the area of error correcting codes is whether, for a fixed
value of $\ell$, there exists an infinite sequence of good cyclic
codes. 

Most evidence, and maybe the prevailing opinion, goes towards the
non-existence of good cyclic codes. Indeed, it was proved by Berman
\cite{B} in 1967 that if $\{n\}$ ranges over integers whose prime factors are 
bounded, and these factors are coprime to the characteristic of the 
underlying field $\mathbb{F}_{\ell}$,
then no sequence of cyclic codes of lengths $\{n\}$, is good. Babai, Shpilka and
Stefankovic \cite{BSS} proved that this is also the case if $n$ ranges
over integers such that the primes $p$ dividing $n$ all satisfy
$p\leq n^{\frac{1}{2} - \epsilon}$ for some fixed constant
$\epsilon>0$.  Furthermore, they also showed that there are no good
cyclic codes that are either locally testable or LDPC (``low density
parity check'') codes.
We refer to the book~\cite{mws} of MacWilliams and Sloane and to the
textbook of Roth~\cite{roth} for basic terminology and concepts in
coding theory.

On the other hand, the uncertainty principle is a classical result of
harmonic analysis, which in one form asserts that given a function
$f$, either $f$ or its Fourier transform $\hat{f}$ has large
support. Many variants exist, and we refer to Folland and
Sitaram~\cite{FS} for a survey of the continuous setting. We will
consider the version of the uncertainly principle where
$f:A\rightarrow\mathbb{C}$ is a complex valued function on a finite
group $A$, and even more particularly, when $A$ is the cyclic group
$\mathbb{Z}/p\mathbb{Z}$ of prime order $p$.  In this case, the
uncertainty principle states that for $f \not=0$, we have
\begin{equation} \label{up}
|\supp(f)| + |\supp(\hat{f})| \geq p + 1,
\end{equation}
where $\supp(g)$ is the support of a function (see Meshulam~\cite{M1},
Goldstein, Guralnick, Isaacs~\cite{GGI}, Tao~\cite{T} or \S \ref{s-3}
below).

One can formulate the uncertainty principle for functions from
$A=\mathbb{Z}/p\mathbb{Z}$ to any algebraically closed field $F$ (see
Section~\ref{s-3}). The case of interest to us is when $F$ has
positive characteristic $\ell$, in particular when $\ell=2$. The
inequality (\ref{up}) does not hold in general in this case (see \S
\ref{s-4} below), but we will give some heuristic argument suggesting
that some weaker version may still hold.

We will then show that even a much weaker version of the inequality
(\ref{up}) for $F=\bar{\mathbb{F}}_2$ would suffice to imply the
\underline{existence} of good cyclic codes.  This should come as quite
a surprise, as it goes against the common wisdom in the theory of
error correcting codes.

\subsection*{Acknowledgements}

The authors are grateful to E. Ben-Sasson, B. Poonen, P. Sarnak and
M. Sudan for discussions and suggestions, many of which have been
incorporated into the text.  We thanl F. Voloch for pointing out his
note~\cite{V}. We acknowledge support by the ERC, NSF, ISF, Dr.~Max
R\"{o}ssler, the Walter Haefner Foundation and the ETH Foundation, and
the ETH Institute for Theoretical Studies. EK's work is partially
supported by an DFG-SNF lead agency program grant (grant
200021L\_153647).

\subsection{Organization of the paper}

This note is arranged as follows:

In $\S$ \ref{s-2}, we describe cyclic codes of length $n$ over the
prime field $\mathbb{F}_{\ell}$ of order $\ell$, as ideals in the
group algebra
$\mathbb{F}_{\ell}[\Zz/n\Zz] \cong \mathbb{F}_{\ell}[x]/(x^n-1)$.  We
then describe the structure and the ideals of
$\mathbb{F}_{\ell}[\Zz/p\Zz]$ when $n=p$ is a prime, and express the
dimension and the distance of such an ideal in terms of this data
(using in particular the multiplicative order of $\ell$ modulo $p$).

In $\S$ \ref{s-3}, we formulate the uncertainty principle for
functions $f:\Zz/p\Zz \rightarrow \mathbb{C}$.  To illustrate the
connection with cyclic codes, we show how this uncertainty principle
implies the existence of good cyclic codes over $\mathbb{C}$ -- the
examples we recover are the well-known Reed-Solomon codes over
$\mathbb{C}$.  This is of course not the end of the story, as one
wants such codes over finite fields.

In $\S$ \ref{s-4}, we formulate a few variants of the uncertainty
principle over various fields. We present a proof of the uncertainty
principle for any field of characteristic zero, following~\cite{GGI}.
Afterwards, we present some counter-examples to a naive generalization
of the uncertainty principle to finite fields.

In $\S$ \ref{s-5}, we propose a weaker version of uncertainty
principle, and show how this weaker version implies the existence of
good cyclic codes.  In $\S$ \ref{s-6}, we present some heuristics,
both for this weak uncertainty principle and for the existence of good
cyclic codes. 

We conclude with an Appendix that explains that the uncertainty
principle for $\Zz/p\Zz$ is equivalent to an old result of Chebotarev.


\section{Cyclic codes} \label{s-2}

\subsection{Introduction}

The following is a long standing open problem.

\begin{prb} \label{no-good} 
Are there good cyclic codes over a fixed finite field $F$?
\end{prb}

This was asked by MacWilliams and Sloane~\cite[Problem 9.2,
p. 270]{mws}.
See also \cite{MPW} who attribute the problem to \cite{AMS}. It seems that the common belief is that there are no 
such codes and there are a number of results in support of such a conjecture.

For instance, the most commonly used cyclic codes are the long BCH
codes (see~\cite[\S 5.6]{roth} for definition and background of BCH
codes), and Lin and Weldon \cite{LW} proved that these codes are not
good.

Partial results toward the conjecture were obtained by Berman \cite{B}
in 1967 and
by Babai, Shpilka and Stefankovic \cite{BSS} in 2005. We state their
results formally:

\begin{thm}[Berman]
  Let $F$ be a finite field of order $\ell$, and $(C_t)_t$ a family of
  $[n_t,k_t,d_t]_F$-cyclic codes such that there exists some real
  number $c>0$ with $\frac{k_t}{n_t} \geq c$ for all $t$.  Assume
  furthermore that there exists $\beta\geq 1$ such that all primes
  dividing $n_t$ are coprime to $\ell$ and at most $\beta$.  Then
  there exists an integer $m$, depending on $\ell$ and $\beta$, such
  that $d_t \leq m$.
  In particular, this family is not a good family of codes.
\end{thm}

\begin{thm}[Babai-Shpilka-Stefankovic]
  Let $F$ be a finite field, and let $(C_t)_t$ be
  a family of $[n_t,k_t,d_t]_F$-cyclic codes over $F$.  Assume that
  there exists $\delta>0$, independent of $t$, such that for every $t$
  and for every prime $p$ dividing $n_t$, we have
  $p < n_t^{1/2-\delta}$. Then the family $(C_t)_t$ is \emph{not} a
  good family of codes over $F$.
\end{thm}

There are other results which give some support to a negative answer to Problem \ref{no-good}, for example:
\begin{thm}[Babai-Shpilka-Stefankovic] \label{ldpc} Let $F$ be a
  finite field. Then:
\begin{itemize}
\item There are no good cyclic LDPC \emph{(}low density parity
  check\emph{)} codes over $F$;
\item There are no good cyclic locally testable codes over $F$.
\end{itemize}
\end{thm}

We refer to~\cite[Ch. 47]{mckay} for the definition of LDPC codes, and
to~\cite{gsudan} for locally testable codes; these are important
concepts in coding theory in recent years.

Let $F$ be any field. The key to the investigation of cyclic codes
over $F$ is their description in algebraic terms using the polynomial
ring $F[X]$.

\begin{pro}
Let $n\geq 1$ be an integer. Under the isomorphism
$$
(a_0,\ldots,a_{n-1})\mapsto a_0+a_1X+\cdots+a_{n-1}X^{n-1}
$$
between $F^n$ and the ring $R=F[X]/(X^n-1)$, a subspace $C\subset R$
is a cyclic code over $F$ if and only if $C$ is an ideal of $R$.
\end{pro}

\begin{proof}
  Indeed, an $F$-vector subspace of $R$ is a cyclic code if and only
  if $XP\in C$ for any $P\in C$, which is equivalent to asking that
  $C$ be an ideal of $R$.
\end{proof}

It will also often be convenient to identify the ring $R$ with the
subspace of polynomials $P\in F[X]$ of degree less than $n$.

\subsection{Describing the ideals of $R=F[X]/(X^n-1)$}\label{ssec-ideals}

If we specialize to the case where $n=p$ is a prime number, we can
describe $R$ and its ideals in quite concrete and well-known terms:

\begin{pro} \label{decom} Let $p$ be a prime number different from the
  characteristic $\charac(F)$ of $F$. Then:
\begin{enumerate}
\item The ring $R=F[X]/(X^p-1)$ is a direct sum of finite extensions
  of $F$; these finite extensions are in one to one correspondence
  with the irreducible factors of the polynomial $X^p-1 \in F[X]$.
\item If $x^p-1$ splits in linear factors in $F[x]$ (e.g. if $F$ is
  algebraically closed), then $R$ is isomorphic to $F^p$ as a ring;
\item Assume that $F=\mathbb{F}_{\ell}$ is a finite field of order
  $\ell$.  Let $r=\ord_p(\ell)$, i.e. the order of $\ell$ as an
  element of the multiplicative group
  $(\mathbb{Z}/p\mathbb{Z})^* = \mathbb{F}_p^*$.  Denote $s=(p-1)/r$.
  Then
$$
R=\mathbb{F}_{\ell}[X]/(X^p-1)\cong \mathbb{F}_{\ell} \oplus
(\mathbb{F}_{\ell^r})^s
$$
i.e., it is isomorphic as a ring to a direct sum of
$\mathbb{F}_{\ell}$ and $s$ copies of the extension
$\mathbb{F}_{\ell^r}$ of $\mathbb{F}_{\ell}$.
\end{enumerate}
\end{pro}

\begin{proof}
\begin{enumerate}
\item As $p \ne \charac(F)$, the polynomial $X^p-1$ is separable in
  $F[X]$ and hence factors as a product of distinct irreducible
  polynomials $\prod_{i=0}^s g_i$, where we put $g_0=X-1$.  It then
  follows from the Chinese Remainder Theorem that
$$
R \cong \bigoplus_{i=0}^s F[X]/(g_i).
$$
Since $g_i$ is irreducible, each quotient ring $F[X]/(g_i)$ is a field
extension of $F$ of degree $\deg(g_i)$.

\item By assumption, $X^p-1 = \prod_{i=0}^{p-1}(X-\mu_i)$, where
  $\mu_i$ runs over the $p$-th roots of unity in $F$. Since
  $F[X]/(X-\alpha) \cong F$, we get an isomorphism
$$
R \cong \bigoplus_{i=0}^{p-1} F[X]/(X-\mu_i) \cong F^p.
$$

\item Since $\mathbb{F}_{p}^*$ is a cyclic group of order $p-1$, the
  order $r$ of $\ell$ modulo $p$ divides $p-1$, and hence $s=(p-1)/r$
  is an integer.
\par
We have $\ell^r \equiv 1(\mbox{mod }p)$ and $\mathbb{F}_{\ell^r}^*$ is
a cyclic group of order $\ell^r-1$, hence the field extension
$\mathbb{F}_{\ell^r}$ of $\mathbb{F}_{\ell}$ contains an element of
order $p$, and is the smallest extension with this property. In fact,
the field $\mathbb{F}_{\ell^r}$ contains all the $p$-th roots of
unity, i.e. $\mathbb{F}_{\ell^r}$ is the splitting field of the
polynomial $X^p-1$. For every $p$-th root of unity $\mu$, the
extension $\mathbb{F}_{\ell}[\mu]$ is equal to $\mathbb{F}_{\ell^r}$
(in a fixed algebraic closure of $\mathbb{F}_{\ell}$).  This shows
that all the irreducible factors $g_i$ of $X^p-1$, with the exception
of $X-1$, are of degree $r$.  Hence
$$
R \cong \mathbb{F}_{\ell} \oplus (\mathbb{F}_{\ell^r})^s.
$$
\end{enumerate}
\end{proof}

We can now describe the ideals of $R$.  Since $R$ is a direct sum of
fields, every ideal in $R$ is the direct sum of a certain subset of
these fields.  If $F$ is algebraically closed, for instance, we see
that $R$ has $\binom{p}{i}$ distinct ideals of dimension $i$, for
every $0 \leq i \leq p$, and a total of $2^p$ ideals.

If $F=\mathbb{F}_{\ell}$ where $\ell$ is the power of a prime number,
let $r$ be the order of $\ell$ modulo $p$ and $s = \frac{p-1}{r}$ as
in the proposition.  In the special case $r=1$, namely when
$p\mid \ell-1$, the polynomial $X^p-1$ splits completely in
$\mathbb{F}_{\ell}[X]$ and the ideals are exactly the same as those in
the algebraically closed case.  

Now assume that $r>1$, which is the case we are most interested in
since we will consider a fixed value of $\ell$ as $p$ tends to
$\infty$. Then $R$ has $\binom{s}{i}$ ideals of dimension $ir$ and
$\binom{s}{i}$ ideals of dimension $ir+1$ for all integers $i$ with
$0 \leq i \leq s$. Hence the total number of ideals in $R$ is
$2^{s+1}$.

We note that $r \geq \log_{\ell}(p+1)$, and hence
$s \leq \frac{p-1}{\log_{\ell}(p+1)}$.  

There are two extreme cases which are worth singling out, although
whether they actually occur is somewhat conjectural:
\begin{description}
\item[(a)] Assume that $\ell$ is a primitive root mod $p$, i.e. $\ell$
  generates the cyclic group $(\Zz/p\Zz)^*$.  Then $r=p-1$ and so $s=1$, i.e.
  $R \cong \mathbb{F}_l \oplus \mathbb{F}_{l^{p-1}}$ and $R$ has only
  two non-trivial ideals.
\item[(b)] Assume that $\ell=2$ and that $p$ is a Mersenne prime,
  namely $p=2^m-1$ for some $m\geq 2$. Then we have $r=m=\log_2(p+1)$
  and $s=\frac{p-1}{\log_2(p+1)}$; in this case, $R$ has the
  ``maximal'' possible number of ideals
  $2^{\frac{p-1}{\log_2(p+1)}+1}$.
\end{description}

We stated that it is not known if these cases occur infinitely
often. Indeed, it is a very famous conjecture of Artin (see Moree's
survey~\cite{moree}) that, for a given prime number $\ell$, there
exist infinitely many primes $p$ such that $\ell$ is a primitive root
modulo $p$.  The validity of this conjecture is extremely likely,
since it was shown by Hooley~\cite{hooley} to follow from a suitable
form of the Generalized Riemann Hypothesis. Moreover, although it not
known to hold for any concrete single prime $\ell$,
Heath-Brown~\cite{HB} has shown that it holds for all but at most two
(unspecified) prime numbers.
\par
On the other hand, although it is expected that there are infinitely
many Mersenne primes, very little is known about this question, or
about small values of $\ord_{p}(2)$ in general, even assuming such
conjectures as the Generalized Riemann Hypothesis (see however
Lemma~\ref{lm-split}).

The most convenient analytic criterion to find primes with $\ord_p(\ell)$
under control is the following elementary fact:

\begin{lem}\label{lm-chebo}
  Let $\ell$, $q$ and $ p$ be different primes. If $p$ is totally
  split in the extension
  $K_{q,\ell}=\mathbb{Q}(e^{2i\pi/q},\sqrt[q]{\ell})$, then $p$ is
  congruent to $1$ modulo $q$ and the order of $\ell$ modulo $p$
  divides $(p-1)/q$, in particular $\ord_p(\ell)<p/q$.
\end{lem}

\begin{proof}
  Let $\mathcal{O}$ be the ring of integers of $K_{q,\ell}$. If $p$ is
  totally split in $K_{q,\ell}$, then the quotient ring
  $\mathcal{O}/p\mathcal{O}$ is a product of copies of the field
  $\mathbb{F}_p$. So $\mathbb{F}_p$ contains the $q$-th roots of
  unity (in particular, $q \mid p-1$) and the $q$-th roots of
  $\ell$. So $\ell$ is an $q$-th power in $\mathbb{F}_p$, which means
  that $\ord_p(\ell)$ divides $(p-1)/q$.
\end{proof}

Note that as an application of Chebotarev's density Theorem \cite[Th.~13.4]{neukirch}, 
for any primes $q, \ell$, there exists infinitely many primes which 
totally splits in $K_{q,\ell}$. 

To summarize the discussion: the ideals of $R$ and their dimensions
can be easily described, although the existence of certain
configurations might be subject to the truth of certain arithmetic
conjectures.

It is more complicated to evaluate the distance of ideals of $R$ when
interpreted as cyclic codes. For this we will use the Fourier
transform and the uncertainty principle in the next section. We begin
first with a general lemma.

\begin{lem} \label{lem-dim-roots} Let $p$ be a prime. For any
  polynomial $f\in F[X]$, let $I_f$ be the ideal generated by the
  image of $f$ in $R=F[X]/(X^p-1)$ and let $g=\gcd(f,X^p-1)$.
\begin{enumerate}
\item We have $I_f=I_g$, i.e. the ideal
  generated by $f$ is the same as the ideal generated by the greatest
  common divisor of $f$ and $X^p-1$.
\item We have
$$
\dim I_f = \dim I_g= p-\deg(g)
$$
\end{enumerate}
\end{lem}

\begin{proof}
  (a) We obviously have $ \gcd(f,X^p-1)\mid f$ in $F[X]$, and since
  $F[X]$ is a principal ideal domain, there exist polynomials $h_1$
  and $h_2$ in $F[X]$ such that $\gcd(f,X^p-1) = h_1f +
  h_2(X^p-1)$. Hence we get $f \mid \gcd(f,X^p-1)$ in $R$, which
  proves claim (a).
\par
(b) The first equality follows from (a). For the second equality, it
suffices to note that, by euclidean division by the polynomial
$(X^p-1)/g$ of degree $d=p-\deg(g)$, the elements
$\{X^i\cdot f\,|\,i=0,1,\ldots,d-1\}$ form a basis of $I_f$.
\end{proof}

For later reference, we will denote $\Z(f)=\deg(\gcd(f,X^p-1))$ for any
polynomial $f\in F[X]$ and any prime $p$. If $F$ has characteristic
different from $p$, then $X^p-1$ is a separable polynomial, and in
that case, the integer $\Z(f)$ is therefore the number of $p$-th roots
of unity $\xi$, in an algebraic closure of $F$, such that
$f(\xi)=0$. This interpretation will be very useful as we now turn to
the uncertainty principle...


\section{The uncertainty principle over $\Cc$} \label{s-3}

\subsection{The Fourier transform on finite abelian groups}

Let $A$ be a finite abelian group. The dual group $\widehat{A}$ of $A$
is the group of all homomorphisms $A\to \mathbb{S}^1$, where
$\mathbb{S}^1$ is the group of complex numbers of modulus $1$. The
product on $\widehat{A}$ is the pointwise multiplication of
functions. The dual group is also a finite abelian group, in fact it
is isomorphic to $A$ (non-canonically).
\par
The \emph{Fourier transform} on $A$ is a linear map from the space
$L^2(A)=\Cc^{A}$ of complex-valued functions on $A$ to the analogue
space $L^2(\widehat{A})$ of complex-valued functions on the dual
group. For a function $f\colon A\to\Cc$, its Fourier transform
$\widehat{f}\colon \widehat{A}\to\Cc$ is defined by
$$
\widehat{f}(\chi)=\frac{1}{|A|}\sum_{a\in A}f(a)\overline{\chi(a)}
$$
for any $\chi\in\widehat{A}$.
\par
The Fourier transform is also an algebra isomorphism, where $L^2(A)$
is viewed as an algebra with the convolution product
$$
(f_1\star f_2)(x)=\frac{1}{|A|}\sum_{a\in A}f_1(x-a)f_2(a),
$$
and $L^2(\widehat{A})$ has the pointwise product of functions.  In
other words, we have
$$
\widehat{f_1\star f_2}=\widehat{f}_1\cdot \widehat{f}_2.
$$

The connection that we will make with cyclic codes emphasizes the
group algebra of a cyclic group. It is therefore convenient to
interpret the Fourier transform in terms of the group algebra $\Cc[A]$
of the group $A$ instead of $L^2(A)$.

We identify $L^2(A)$ and $\Cc[A]$ by the map
$$
f \mapsto \sum_{a\in A}f(a) a.
$$
Then the Fourier
transform gives an isomorphism 
$$
\Cc[A]\longrightarrow \Cc^{A}
$$
of algebras over $\Cc$, where the image of the standard basis
$\{a\in A\}$ is the basis of characters of the algebra of functions
$\Cc^{|A|}$.



\subsection{The general uncertainty principle for finite abelian
  groups}\label{sec-abup}

For $f\in L^2(A)$, or equivalently $f\in\Cc[A]$, we denote by
$\supp(f)$ the support of $f$, namely the set of $a\in A$ such that
$f(a)\ne 0$.

Intuitively, by``uncertainty principle'', we mean a statement that
asserts that there are no non-zero functions $f$ such that both $f$
and $\widehat{f}$ have ``small'' support (for instance, in the
continuous case, there is no non-zero smooth function with compact
support whose Fourier transform is also compactly supported).  There
are many variants of this principle. One well-known elementary
``uncertainty principle'' version, valid for all finite abelian
groups, is the following result of Donoho and Stark~\cite[\S
2]{donoho-stark}:

\begin{pro}[Uncertainty principle]\label{pr-uncertain}
  Let $A$ be a finite abelian group and let $f\not=0$ be a function
  from $A$ to $\Cc$. Then we have
\begin{equation}
|\supp(f)|\cdot |\supp(\hat{f})| \geq |A|
\end{equation} 
\end{pro}

We present the proof of this fact from \cite{GGI}, which fits well
with our point of view of working with group algebras.  For other
proofs and generalizations, we refer to the papers \cite{M2},
\cite{M3} and \cite{T}, as well as to the references contained in
those articles.

\begin{proof} 
  We view $f$ as an element of the group algebra $\Cc[A]$, which is
  commutative. Let $I=(f)$ be the principal ideal generated by
  $f$. Using the isomorphism $\Cc[A]\simeq \Cc^{A}$ given by the
  Fourier transform, as we recalled above, the ideal $I$ corresponds
  to the principal ideal in $\Cc^{A}$ generated by the Fourier
  transform of $f$. This ideal is simply 
$$
\prod_{\widehat{f}(x)\not=0} \Cc\subset \Cc^{A}.
$$
In particular, the dimension $r$ of $I$, as a $\Cc$-vector space, is
the cardinality of the support of $\widehat{f}$. Since the elements
$a\cdot f$ for $a\in A$ span $I$ as $\Cc$-vector space, there exist
$r$ elements $a_1$, \ldots, $a_r$ such that $I$ is the span of
$a_1\cdot f$, \ldots, $a_r\cdot f$.
\par
For any $a\in A\subset \Cc[A]$, the support of $a\cdot f$ is
$a\cdot \supp(f)$. Since $f\ne0$, its support is not empty, hence for
any $x\in A$, we can find some element $a\in A\subset \Cc[A]$ such
that $x\in \supp(a\cdot f)$.
\par
We then have
$$
A=\bigcup_{a\in A} \supp(a\cdot f)\subset \bigcup_{i=1}^r
\supp(a_i\cdot f)
$$
which implies that
$$
|A|\leq \sum_{i=1}^r |\supp(a_i\cdot f)| =
r|\supp(f)|=|\supp(\widehat{f})|\cdot |\supp(f)|,
$$
as claimed.
\end{proof}

\subsection{The uncertainty principle for simple cyclic groups}

In the late 1980's, R. Meshulam observed that an old result of
Chebotarev implies a version of the uncertainty principle for cyclic
groups of prime order $p$ that is much stronger than
Proposition~\ref{pr-uncertain}. This strong version has been
rediscovered several times since then, and admits a number of proofs
and generalizations (see for instance, Chebotarev~\cite{C},
Meshulam~\cite{M1,M2,M3}, Goldstein, Guralnick and Isaacs~\cite{GGI},
Tao~\cite{T}, Stevenhagen and Lenstra~\cite{SL}, and the references
therein).

\begin{thm}[Uncertainty principle for cyclic groups
  of prime order]\label{up-abelian} 
  Let $A$ be a cyclic group of prime order $p$, and $f\not=0$ an
  element of $\Cc[A]$. Then
\begin{equation}\label{3}
  |\supp(f)|+|\supp(\widehat{f})| \geq p+1.
\end{equation}
\end{thm}

We will postpone the proof to Section~\ref{sec-abup}, and in the
appendix, we will also explain Meshulam's original observation that
this statement is equivalent to a classical result of Chebotarev about
Vandermonde matrices.

To bring the connection with codes, we will now reformulate this
statement. The group algebra $\Cc[\Zz/p\Zz]$ of the cyclic group of
order $p$ is isomorphic to the quotient algebra $R=\Cc[X]/(X^p-1)$ by
mapping the generator $1$ of $\Zz/p\Zz$ to the image of $X$. The dual
group $\widehat{\Zz/p\Zz}$ is isomorphic to the group $\mu_p(\Cc)$ of
$p$-th roots of unity in $\Cc$, by mapping a character $\chi$ to the
$p$-th root of unity $\chi(1)$. The Fourier transform of an element
$f\in R$, represented as the image of a polynomial
\begin{equation}\label{eq-repr-pol}
  f=a_0+a_1X+\cdots+a_{p-1}X^{p-1}
\end{equation}
is then identified with the function defined on $p$-th roots of unity
by
$$
\widehat{f}(\xi)=\frac{1}{p}\sum_{i=0}^{p-1} a_i\xi^{-i}.
$$
In other words, $\widehat{f}$ is the evaluation of the representing
polynomial~(\ref{eq-repr-pol}) at roots of unity.


With this notation, recalling the definition $\Z(f)=\deg(\gcd(f,X^p-1))$ and
the fact that this is number of zeros of $f$ among $p$-th roots of
unity, the uncertainty principle of Theorem \ref{up-abelian} gets the
following form:

\begin{thm} \label{up-abelian2} Let $p$ be a prime. For any polynomial
$$
f = \sum_{i=0}^{p-1} a_iX^i \in \mathbb{C}[X]
$$
of degree $<p$, let $\wt(f) = |\{i| a_i \ne 0 \}|$ and let
$\Z(f) = |\{\mu\in\mu_p(\Cc)| f(\mu) =0 \}|$, i.e.  the number of
$p$-th roots of unity of $f$ which are also roots of $f$. Then we have
\begin{equation} 
  \label{4} \Z(f) \leq \wt(f) -1.
\end{equation}
\end{thm}

Indeed, by definition, if we view $f$ as an element of
$R=\Cc[\Zz/p\Zz]$, then we have $|\supp(f)|=\wt(f)$ and
$|\supp(\widehat{f})|=p-\Z(f)$, and therefore \eqref{3} and \eqref{4}
are equivalent.

\begin{rem}
  (1) The restriction $\deg(f)<p$ is necessary: the polynomial
  $f=X^p-1$ has $\wt(f)=2$ and $\Z(f)=p$.
\par
(2) The inequality \eqref{4} is best possible. For instance, the
cyclotomic polynomial $f=\frac{X^p-1}{X-1} = 1+X+\ldots+X^{p-1}$
vanishes on all the non-trivial $p$-roots of unity, so
$\Z(f)=p-1=\wt(f)-1$.  Another example is $f=X-1$, in which case we
also obtain $\Z(f)=1=\wt(f)-1$.
\end{rem}

We can now use Lemma~\ref{lem-dim-roots} to obtain another
reformulation of Theorems \ref{up-abelian} and \ref{up-abelian2}.  The
point is that if $f$ is a polynomial in $\Cc[X]$ of degree $<p$,
viewed also as an element of $R$, then by Lemma~\ref{lem-dim-roots}
(2), the dimension of the ideal $I_f$ generated by the image of $f$ in
$R$ satisfies
$$
\dim(I_f)= p - \Z(f).
$$
From Theorem~\ref{up-abelian2}, we get therefore:

\begin{thm}[Uncertainty principle reformulated] \label{up-abelian3}
  For every non-zero polynomial $f \in \mathbb{C}[X]$ of degree $<p$,
  considered as an element of $R=\mathbb{C}[X]/(X^p-1)$, we have:
\begin{equation}\label{eq-up3}
  \wt(f) + \dim(I_f) \geq p + 1
\end{equation} 
when $I_f=(f)$ is the ideal of $R$ generated by the image of $f$.
\end{thm}

We conclude this section by showing how this interpretation of the
uncertainty principle gives a good family of cyclic codes over
$\mathbb{C}$:

\begin{cor}
  There exists a family of good cyclic codes over $\mathbb{C}$.
\end{cor}

\begin{proof}
  Let $ \xi = e^{\frac{2\pi i}{p}} \in \mathbb{C}$, and define
$$
f = \prod_{i=1}^{\frac{p-1}{2}} (X-\xi^i) \in \mathbb{C}[X].
$$
Since $f|(X^p-1)$, we have $\dim(I_f) = p - \deg(f) = \frac{p+1}{2}$
by Lemma~\ref{lem-dim-roots} (2).
\par
Let then $h\not=0$ be an element of $I_f$. We then have
$\dim(I_h)\leq \dim(I_f)$, so that
$$
\wt(h) \geq p+1 - \dim(I_h)\geq p+1 - \dim(I_f)= \frac{p+1}{2}
$$
by Theorem \ref{up-abelian3}. The ideal $C_p=I_f$ is therefore a
$[p,\frac{p+1}{2},\frac{p+1}{2}]_{\mathbb{C}}$-cyclic code, and the
family $\{C_p\}_{p\text{ prime}}$ is a good family of cyclic codes.
\end{proof}

The codes we have ``found'' in this proof are special cases of the
famous Reed-Solomon codes (see, e.g.,~\cite[\S 5.2]{roth}).

\section{Uncertainty principle for general fields} \label{s-4}

\subsection{General statements}

The formulation of the uncertainty principle in
Theorems~\ref{up-abelian2}, in the form of the inequality~(\ref{4})
and in Theorem \ref{up-abelian3}, through~(\ref{eq-up3}), make sense
for all fields.  As we will see later, these statements are not true
in such generality, but they might be true, and useful, in some weaker
form. For this reason, we make the following definition.

\begin{dfn}
  Let $F$ be a field, $p$ a prime number and $R=F[X]/(X^p-1)$. For
  $f\in R$, represented by a polynomial of degree $< p$, we denote by
  $I_f$ the ideal generated by $f$ in $R$, and we denote
$$
\mu_{F,p}(f) = \wt(f) + \dim(I_f).
$$ 
We then define the invariant
$$
\mu_{F,p}=\min\{\mu_{F,p}(f)|0 \ne f\in R\}.
$$
We will sometimes write $\mu(f)$ instead of $\mu_{F,p}(f)$, when the
field and prime involved are clear in context.
\end{dfn}

Here are some simple observations:
\begin{itemize}
\item If $E/F$ is a field extension and $f \in F[X]/(X^p-1)$, then
  $\mu_{F,p}(f)=\mu_{E,p}(f)$ for any prime number $p$.  In
  particular, it follows that $\mu_{E,p} \leq \mu_{F,p}$ for each $p$.
\item For $f = 1+X+\ldots+X^{p-1}$, we have $\wt(f)=p$ and
  $\dim(I_f)=1$. It follows that $\mu_{F,p} \leq p+1$ for any field
  $F$ and any prime $p$.
\item According to the uncertainty principle for $F=\Cc$ (Theorems
  \ref{up-abelian}, \ref{up-abelian2} and \ref{up-abelian3}), we have
  $\mu_{\mathbb{C},p} = p+1$ for every prime $p$.
\end{itemize}

So for any field  we can state the uncertainty principle as follows:

\begin{dfn}[Uncertainty principle]\label{def-up}
  A field $F$ is said to satisfy the uncertainty principle if, for any
  prime number $p$, we have $\mu_{F,p}>p$, or equivalently if
  $\mu_{F,p} = p+1$, for all $p$.
\end{dfn}

As we shall see in \S 4.2, the uncertainty principle does not hold in
general, but let us start with some positive results:

\begin{pro} \label{artin} Let $F=\mathbb{F}_{\ell}$ be the finite
  field of prime order $\ell$ and assume that $\ell$ is a primitive
  root modulo $p$, i.e., that $\ord_p(\ell)=p-1$.  Then
  $\mu_{F,p} = p+1$.
\end{pro}

\begin{proof}
  Let $\xi \ne 1$ be a primitive $p$-th root of unity in
  $\bar{\mathbb{F}}_{\ell}$. As recalled in Section~\ref{ssec-ideals},
  the extension $\mathbb{F}_{\ell}(\xi)/\mathbb{F}_{\ell}$ is then of
  degree $\ord_p(\ell)=p-1$. This implies that the polynomial
  $\frac{X^p-1}{X-1} = 1+X+\ldots+X^{p-1}$ is irreducible over
  $\mathbb{F}_{\ell}$. In particular, for every polynomial
  $f\in \mathbb{F}_{\ell}[X]$ of degree less then $p$, the gcd of $f$
  and $X^p-1$ can only be one of $1$, $X-1$ or $(X^p-1)/(X-1)$. Then
  the dimension $\dim(I_f) = p-\deg(\gcd(f,X^p-1))$ is equal to $p$,
  $p-1$ or $1$, respectively (Lemma~\ref{lem-dim-roots} (2)).

  We consider each case in turn and show that $\mu(f)\geq p+1$ in any
  case. If $\dim(I_f)=p$, then since $\wt(f) \geq 1$ (because
  $f\ne 0$), we get $\mu(f) \geq p+1$.  If $\dim(I_f) = p-1$, then we
  have $\gcd(f,X^p-1) = X-1$, so $X-1\mid f$. Since the only non-zero
  polynomials of weight $1$ are monomials $cX^i$ with $c\not=0$, and
  $X-1\nmid cX^i$ for $0\leq i<p$, we must have $\wt(f)\geq 2$, and
  therefore $\mu(f)\geq p-1+2=p+1$.  Finally, if $\dim(I_f)=1$, then
  we have $f =c \sum_{i=0}^{p-1}X^i$ for some $c\not=0$, and then
  $\wt(f)=p$ and $\mu(f)=p+1$.
\end{proof}

Another case is the following claim (which appears also in \cite[Lemma~2]{F} and \cite[Lemma~6.5]{GGI}), 
that we will use later:

\begin{pro} \label{p,p} Let $p$ be a prime and let $F$ be a field of
  characteristic $p$. Then we have $\mu_{F,p} = p + 1$.
\end{pro}

\begin{proof}
  By Lemma~\ref{lem-dim-roots} (2), we need to show that for any
  $0\ne f \in F[X]/(X^p-1)$, we have
$$
\wt(f) > p-\dim(I_f)=\deg(\gcd(f,X^p-1)).
$$ 
Since $F$ has characteristic $p$, we have $X^p-1=(X-1)^p$, which means
that there exists some integer $m$ with $0\leq m<p$ such that
$\gcd(f,X^p-1)= (X-1)^m$.  So we need to prove that for a polynomial
$f$ with $(X-1)^m|f$, we have $\wt(f)> m$.

We proceed by induction on $\deg(f)<p$. In the base case $\deg(f)=0$,
we have $f=c \ne 0$. Then $X-1 \nmid f$, so that $m=0$ and
$\wt(f)=1>m$, as claimed. 

Now assume that the property is valid for all polynomials of degree
$<\deg(f)$ and that $(X-1)^m|f$.  If $f(0)=0$, we deduce that
$(X-1)^m|f(X)/X$, hence by induction we obtain
$m < \wt(f/X) = \wt(f)$.  If $f(0)\ne 0$, on the other hand, then we
consider the derivative $f'$ of $f$. From $(X-1)^m\mid f$, it follows
that $(X-1)^{m-1}\mid f'$: indeed, writing $f=f_1(X-1)^m$ and
differentiating, we get $f'=f'_1(X-1)^m+mf_1(X-1)^{m-1}$, which is
divisible by $(X-1)^{m-1}$. By induction, we therefore get
$\wt(f')>m-1$. But then, since $f(0)\ne 0$ and $m<p$, we have
$\wt(f) = \wt(f')+1 > m$, as needed.
\end{proof}

\subsection{Fields of characteristic zero}

We will now present a proof (following \cite{GGI}) of the uncertainty
principle for any field $F$ of characteristic zero.  Note that
Theorems \ref{up-abelian}, \ref{up-abelian2} and \ref{up-abelian3} are
special cases of this result, where the field is $\mathbb{C}$.  Since
it is elementary that we need only prove the uncertainty principle for
finitely generated fields $F$, and since such a field $F$ of
characteristic $0$ can be embedded into $\mathbb{C}$, we could simply
deduce the result from the case of $\Cc$. We give a complete proof
anyway.

The next lemma is the key step in the proof.

\begin{lem}[Specialization] \label{spec} Let $p$ be a prime, $F$ a
  field of characteristic $0$, and
$$
f = \sum_{i=0}^{p-1} a_i X^i
$$
a non-zero element of $R=F[X]/(X^p-1)$.  Then for every prime number
$q$, there exists a field $E$ of characteristic $q$ and a polynomial
$\tilde{f}\in E[X]/(X^p-1)$ such that $\wt(\tilde{f})\leq \wt(f)$ and
$\dim_E(I_{\tilde{f}}) \leq \dim_F(I_f)$.
\end{lem}

\begin{proof}[Sketch of the proof:]
\phantom{c}  
\par
\begin{enumerate}
\item Since $\charac(F) = 0$, the field $\mathbb{Q}$ is a subfield of
  $F$. Let $A = \mathbb{Q}[a_0,\ldots,a_{p-1}]$, which is a
  $\mathbb{Q}$-subalgebra of $F$.  By Hilbert's Nullstellensatz, the
  homomorphisms $\phi\colon A\to \bar{\mathbb{Q}}$,
  where $\bar{\mathbb{Q}}$ is the algebraic closure of $\mathbb{Q}$, 
  separate the points of $A$, and therefore 
  there exists a morphism $\phi : A \rightarrow \bar{\mathbb{Q}}$,
    such that $\phi(a_i) \ne 0$ for every $i$, with $0\leq i\leq p-1$,
  such that $a_i\not=0$.  Let $K_1$ be the number field (a finite
  extension of $\mathbb{Q}$) generated by the image of $\phi$ and
  $f_1$ the polynomial
$$
f_1 = \sum_{i=0}^{p-1}\phi(a_i)X^i\in K_1[X].
$$
Then by the definition of $K_1$, we have $\wt(f_1)=\wt(f)$. Moreover,
$\phi$ induces an isomorphism between the $p$-th roots of unity in
$\bar{K}$ and those in $\bar{\mathbb{Q}}$, so that $\Z(f)=\Z(f_1)$
also. This means that we may replace $K$ and $f$ by $K_1$ and $f_1$,
and reduce to the case where $K$ is a number field.
\item Let $\mathcal{O}_K$ be the ring of integers of $K$, and
  $\mathfrak{m}$ a maximal ideal in $\mathcal{O}_K$ that contains
  $q \in \mathbb{Z} \subset \mathcal{O}_K$.  Then
  $E=\mathcal{O}_K/\mathfrak{m}$ is a finite field of characteristic
  $q$.
\item Let $t\in \mathcal{O}_K$ be a non-zero integer such that
  $ta_i \in \mathcal{O}_K$ for all $i$, and such that there exists
  some $i$ such that $ta_i\notin \mathfrak{m}$ (this exists because
  not all $a_i$ are zero). Then, if $\tilde{f}$ is the image of $tf$
  under the reduction map from $\mathcal{O}_K$ to $E$, we have
  $\tilde{f}\ne 0$ in $E[X]$, and $\tilde{f}$ is a polynomial of
  degree $ < p$.
\item By construction, we have $\wt(\tilde{f})\leq \wt(f)$. On the
  other hand, we get
\begin{align*}
  \dim_F I_f \geq \dim_F I_{tf} &=p-\deg(\gcd(tf,X^p-1)) \\
             &\geq
               p-\deg(\gcd(\tilde{f},X^p-1))=\dim_{E} I_{\tilde{f}}.
\end{align*}
\end{enumerate}
\end{proof}

\begin{thm} \label{up-0} For every field $F$ of characteristic $0$ and
  every prime $p$, we have $\mu_{F,p} = p+1$, i.e., the uncertainty
  principle is true over any field of characteristic $0$.
\end{thm}

\begin{proof}
  Let $F$ be a field of characteristic zero, and let $p$ be a
  prime. Let $f \in F[X]/(X^p-1)$ be non-zero. By the Specialization
  Lemma \ref{spec} with $q=p$, there exists a field $E$ of
  characteristic $p$ and a non-zero element
  $\tilde{f} \in E[X]/(X^p-1)$ such that
  $\mu_{E,p}(\tilde{f}) \leq \mu_{F,p}(f)$.  Because $E$ has
  characteristic $p$, Proposition \ref{p,p} implies that
  $\mu_{F,p}(f) \geq \mu_{E,p}(\tilde{f}) > p$. Since this holds for all
  $f$, the result follows.
\end{proof}

\subsection{Counter examples to the uncertainty principle over finite
  fields}

Specific examples of finite fields $F$ for which the uncertainty
principle of Definition~\ref{def-up} does \emph{not} hold over a
finite field $F$ are given in \cite{GGI}. One such example is
$F=\mathbb{F}_2$. If we take $p=7$ and
$f=X^3+X+1 \in \mathbb{F}_2[X]/(X^7-1)$, then we have
$$
X^7-1 = (X-1)(X^3+X^2+1)(X^3+X+1),
$$ 
hence $\dim(I_f) = 4$ while $\wt(f)=3$, so that
$\mu_{\mathbb{F}_2,7} \leq 7$.

The next counter-examples to the naive uncertainty principal for
finite fields were suggested to us by Madhu Sudan.

Let $q< p$ be two different primes, and $r=\ord_p(q)$. Let
$F=\mathbb{F}_{q}$ and $E =\mathbb{F}_{q^r}$, so that $E$ contains all
the $p$-th roots of unity.  Moreover, $E$ is generated as an
$F$-vector space by the $p$-th roots of unity.   We consider the trace polynomial
$$
T= \sum_{i=0}^{r-1} X^{q^i}\in F[X].
$$
A basic but crucial observation is that the function from $E$ to $E$
defined by the trace polynomial $T$ is a surjective $F$-linear map
from $E$ to the subfield $F$, which we denote $\mathrm{tr}$. In
particular, $\mathrm{tr}$ is not identically zero on $E$, and since
the $p$-th roots of unity generate $E$ as $F$-vector space, this means
that $T$ is not identically zero on the $p$-th roots of unity.

By the pigeon-hole principle, there exists some $\alpha \in F$ such
that at least $\frac{p}{q}$ of the $p$-th roots of unity in $E$ are
roots of $T+\alpha$. Let then $f=T+\alpha\in F[X]$.  Then we have
$$
\mu_{F,p}(f) = \wt(f) + \dim_F(I_f) \leq r+1 +
\Bigl(1-\frac{1}{q}\Bigr)p
$$
(using the interpretation of $\dim_F(I_f)$ as the number of roots of
unity where $f$ does not vanish), and consequently
$$
\mu_{E,p} \leq \mu_{F,p}\leq p+1 + r - \frac{p}{q}.
$$
In particular, if $r=\ord_p(q) < \frac{p}{q}$, we obtain a counter
example to the uncertainty principle for the field
$E=\mathbb{F}_{q^r}$.  

There exist infinitely many pairs of primes with this property. For
instance, take $q=2$ and let $p$ be a prime such that the Legendre
symbol $(\tfrac{2}{p})$ is equal to $1$. Then $q=2$ is a square modulo
$p$, which implies that $2^{(p-1)/2}\equiv 1\bmod p$, hence that the
order of $2$ modulo $p$ is $\leq (p-1)/2<p/2=p/q$.

More generally, fix the prime $q$ and take any prime $\ell>q$. By
Lemma~\ref{lm-chebo}, if $p$ is any prime that is totally split in the
Galois extension
$K_{\ell}=\mathbb{Q}(e^{2i\pi/\ell}, \sqrt[\ell]{q})$, we have
$\ord_p(2)\leq (p-1)/\ell<p/q$. It is a well-known consequence of the
Chebotarev density theorem that there are infinitely such primes.



In anticipation of the next section, we note however that, for any
pair $q<p$ with $r<p/q$, it still remains true that
$$
\mu_{F,p}(f)\geq p+1 +r-\frac{p}{q}\geq \frac{p}{2},
$$
or in other words, the uncertainty principle for $f$ does not fail
drastically.

\section{The weak uncertainty principle} \label{s-5} 

\subsection{Statement}\label{sec-wup}

The uncertainty principle in its current version over $\mathbb{C}$
states that for each prime $p$, we have $\mu_{\mathbb{C}}(p) > p$.  We
have seen that this inequality does not always hold if $\mathbb{C}$ is
replaced by any field.  Because of the link with good cyclic codes, we
introduce a weaker version:

\begin{dfn}[Weak uncertainty principle]\label{def-wup}
  Let $\delta$ be a real number such that $0<\delta\leq 1$. We say
  that a field $F$ satisfies the \emph{$\delta$-uncertainty principle}
  for a prime $p$ if
\begin{equation}
  \mu_{F,p} > \delta \cdot p.
\end{equation}
\end{dfn}

This variant of the uncertainty principle is weaker than the one in
the previous section in two respects: the lower bound for $\mu_{F,p}$
is relaxed, and it is stated with respect to an individual prime $p$,
and not all of them.

\begin{exem}
  We first present some finite fields that satisfy the weak
  uncertainty principle for certain primes. Let $\ell$ be a prime
  number, and let $P$ be an infinite set of primes such that $\ell$ is
  a primitive root in $\mathbb{F}_{p}^*$ for all $p\in P$.  As we have
  already mentioned, Artin's Conjecture asserts that such a set $P$
  exists for any prime $\ell$, and Hooley~\cite{hooley} confirmed this
  under a suitable form of the Generalized Riemann Hypothesis. By
  Proposition \ref{artin}, we have $\mu_{\mathbb{F}_{\ell}}(p) > p$,
  for any $p\in P$, and hence the weak uncertainly principle is
  satisfied by the field $\mathbb{F}_{\ell}$ for any prime in $P$.

  This example does not however lead to good cyclic codes. Indeed, if
  we consider proper ideals
  $I_p\subset
  \mathbb{F}_{\ell}[\Zz/p\Zz]=\mathbb{F}_{\ell}[X]/(X^p-1)$ for
  $p\in P$, the fact that $\ell$ is a primitive root modulo $p$ means
  that $I_p$ is generated either by $X-1$ or by $(X^p-1)/(X-1)$. In
  the first case, we have $\dim I_p=p-1$, but the element $X-1$ has
  weight $2$, so that the distance of the code $I_p$ is $2$. In the
  second case, we have $\dim I_p=1$. In either case, the codes
  corresponding to $I_p$ are not good as $p\to+\infty$ in $P$ since
  one of the inequalities in~(\ref{eq-good-code}) fails.
\end{exem}

This example motivates our last variant of the uncertainty principle.

\begin{dfn}[Weak uncertainty principle, 2]\label{def-wup-2}
  Let $\delta$ and $\epsilon$ be real numbers such that
  $0<\delta\leq 1$ and $0<\epsilon<\delta$. We say that a field $F$ of 
  size $\ell$ satisfies the \emph{$(\epsilon,\delta)$-uncertainty principle} 
  if there exists an infinite set of primes $P$ such that, for all primes
  $p\in P$, the two following conditions holds:
\begin{enumerate}
\item We have $\mu_{F,p} > \delta p$,
\item We have $\ord_{p}(\ell)<\epsilon p$.
\end{enumerate}
\end{dfn}

The existence of finite fields $F$ which satisfy such an uncertainty
principle implies the existence of good cyclic codes over $F$:

\begin{thm} \label{good} Let $F=\mathbb{F}_{\ell}$ be a finite field
  prime order $\ell$.  Assume there exist real numbers
  $0<\epsilon<\delta<1$ such that $F$ satisfies the
  $(\epsilon,\delta)$-uncertainty principle.
  Then there exists an infinite family of good cyclic codes over the
  field $F$.
\end{thm}

\begin{proof}
  For each prime $p \in P$, let $I_p \subset F[X]/(X^p-1)$ be a
  non-zero ideal such that
  $$
  \frac{\epsilon p}{2} \leq \dim( I_p) < \epsilon p.
$$
Such an element exists because $r=\ord_{p}(\ell) < \epsilon p$ by
definition, and $R = F[X]/(X^p-1)$ is a sum of ideals of dimension $r$
each, plus a one dimensional ideal, see Proposition~\ref{decom} (3).

For every element $h \in I_p$, we have $I_h \subset I_p $ and
hence $ \dim(I_h) \leq \dim( I_p)$.  From the weak uncertainty
inequality that we assume, we get
$$
\wt(h)=|\supp(h)| > \delta p - \dim(I_h) \geq \delta p - \dim(I_p) >
(\delta - \epsilon)p.
$$

The cyclic code $I_p$ has length $p$; the last computation shows that
its distance is $\geq (\delta-\epsilon)p$, and its dimension is
$\geq \epsilon p/2$. Hence by definition (see~(\ref{eq-good-code})),
the sequence $(I_p)_{p\in P}$ is an infinite sequence of good cyclic
codes over $F$.

\end{proof}

Generally speaking, condition (1) in Definition~\ref{def-wup-2}
ensures that we can find ideals with ``large'' distance, while
condition (2) is used to show the existence of such ideals with
``large'' dimension.

\begin{rem}
  Our proof shows that any choice of ideal $I_p$, such that
  $ \frac{\epsilon}{2} p \leq \dim(I_p) < \epsilon p$ will give a good
  code.  There are many possibilities for such ideals. This suggests
  that a randomized process might be used to prove existence of cyclic
  good codes even under a weaker uncertainty principle.
\end{rem}

\subsection{A uniform weak uncertainty principle does not hold}

It is only natural to ask (and maybe hope) that a uniform weak
uncertainty principle, uniform with respect to $\delta$, should hold
for all finite fields, or in other words, to ask whether there exists
$\delta > 0$ such that $\mu_{F,p} > \delta p$ for any finite field $F$
and any prime $p$.

We will show -- following an argument of Eli Ben-Sasson -- that,
assuming the existence of infinitely many Mersenne
primes, 
this is not the case. 



 

\begin{pro}[No uniform weak uncertainty principle] \label{pro-uup}
  Assume that there exist infinitely many Mersenne
  primes.  
  Then, for any $\delta>0$, there exists a finite field $F$ and a
  prime number $p$ such that $\mu_{F,p} \leq \delta p$.
\end{pro}

For the proof, we will use the following result of Ore \cite{O}:

\begin{lem}[Ore] \label{ore} Let $q$ be a prime number and $n\geq
  1$. Let $F=\mathbb{F}_{q^n}$, and view $F$ as an
  $\mathbb{F}_q$-vector space of dimension $n$. For every integer
  $k\leq n$ and every $\mathbb{F}_q$-affine subspace $A\subset F$ of
  dimension $k$, the polynomial
$$
f_A = \prod_{a \in A}(X-a)
$$
satisfies
$$
f_A= \alpha + \sum_{i=0}^k \alpha_i X^{q^i}
$$ 
where $\alpha$ and $\alpha_i$ are elements of $F$. In particular, we
have $\wt(f_A)\leq k+2$.
\end{lem}

\begin{proof}
  It is easy to see that it suffices to consider the case where $A$ is
  a vector subspace of dimension $k$. Then $f_A$ is a separable
  polynomial whose roots form an additive subgroup of $F$. This
  implies that $f_A$ is an \emph{additive polynomial}
  (see~\cite[Th. 1.2.1]{goss}), which is necessarily of the desired
  form (with $\alpha=0$ in that case) by~\cite[Prop. 1.1.5]{goss}.
\end{proof}

\begin{rem}
  In general, if $K$ is any field, an \emph{additive polynomial}
  $f\in K[X]$ is a polynomial such that $f(x+y)=f(x)+f(y)$ for any $x$
  and $y$ in $K$. If $K$ has characteristic zero, it is easy to check
  that $f$ is necessarily of the form $f=aX$ for some $a\in K$, but
  this is not so in characteristic $p>0$, since any monomial $X^{p^i}$
  is then an additive polynomial. The result we used is that any
  additive polynomial is a linear combination of these monomials.
\end{rem}

\begin{proof}[Proof of Proposition \ref{pro-uup}]
  Let $q=2$ and let $p=2^n-1$ be a Mersenne prime, so that
  $n = \ord_p(2)$. Let $F=\mathbb{F}_{2^n}$. Then the non-zero
  elements of $F$ are precisely the $p$-th roots of unity.

  We view $F$ as an $n$-dimensional vector space over $\mathbb{F}_2$,
  and fix a basis $e_1$, \ldots, $e_n$. Let $k$ be an integer
  parameter such that $1\leq k<n$.
\par
There exist disjoint affine subspaces $A_1$, \ldots, $A_k$ in $F$,
none of which contains $0$, with $\dim(A_i)=n-i$ (for instance, we
could take $A_i$ to be the subspace defined by  the equations
$$
A_i=\{x\in F\,\mid\, x_1=\cdots=x_{i-1}=0,\quad x_i=1\},
$$
where $(x_1,\ldots,x_n)$ are the coordinates of an element $x$ of $F$
with respect to the chosen basis $(e_1,\ldots,e_n)$).
\par
The disjoint union of the subspaces $A_i$ has cardinality
$$
\Bigl|\bigcup_{1\leq i\leq
  k}A_i\Bigr|=\sum_{i=1}^k2^{n-i}=2^n\Bigl(1-\frac{1}{2^{k}}\Bigr).
$$
Thus if we denote by $f_i$ the polynomial associated to $A_i$ as in
Lemma~\ref{ore}, and put
$$
f=\prod_{i=1}^k f_i\in \mathbb{F}[X],
$$
then we have
$$
\deg(f)=\sum_{i=1}^k\deg(f_i)=\Bigl|\bigcup_{1\leq i\leq
  k}A_i\Bigr|=2^n\Bigl(1-\frac{1}{2^{k}}\Bigr)<2^n-1=p
$$
since $1\leq k<n$ and
$$
\wt(f) \leq \prod_{i=1}^k \wt(f_i) \leq \prod_{i=1}^k (n-i+2) \leq
(n+1)^k.
$$
Since $\gcd(f,X^p-1)=f$, we have
$$
\dim(I_f)=p-\deg(\gcd(f,X^p-1))=p-\deg(f)=2^{n-k}-1\leq \frac{p}{2^k}.
$$
\par
Let $\delta>0$ be any given real number.  Take some integer $k\geq 1$
such that $\frac{1}{2^k}\leq \frac{\delta}{2}$. By the assumption that
there exist infinitely many Mersenne primes, we can find a prime
$p=2^n-1$ for which $n>k$ and
$$
(n+1)^k \leq \frac{\delta}{2} p.
$$
Then using the polynomial $f$ obtained as above for these parameters
$p=2^n-1$ and $k$, we get
$$
\mu_{F,p}\leq \wt(f) + \dim(I_f) \leq (n+1)^k + \frac{p}{2^k}\leq
\frac{\delta}{2} p + \frac{\delta}{2}p = \delta p,
$$
and therefore $\mu_{F,p}\leq \delta p$.
\end{proof}

It is important to notice that this counter-example does \emph{not}
show that $\mathbb{F}_2$ does not satisfy the $\delta$-uncertainty
principle for the prime $p$, since the polynomials $f_i$ and $f$ do
not usually belong to $\mathbb{F}_2[X]$.  Furthermore, as the
underlying field depends on the primes $p$, this counter example is
not really relevant to our search of families of cyclic good codes,
since in such a family we need to work with a fixed underlying field
while in the last example, the size of $F$ grows to infinity.

\section{Why good cyclic codes should exist}\label{s-6}

\subsection{Preliminaries}

In this section, we describe some heuristic arguments that all point
in the direction of the existence of families of good cyclic codes,
and of the weak uncertainty principle according to
Definition~\ref{def-wup-2}.

In both arguments, the main unproved claim is that for a polynomial of
degree $<p$, the property of being ``sparse'' (i.e., of having small
weight $\wt(f)$) and of vanishing on many roots of unity should be
roughly independent. The following result is then relevant.

\begin{lem}\label{lm-entropy}
  Let $\delta$ be a fixed real number with $0<\delta<1/2$. Let
  $S_{\delta}$ be the set of polynomials $f$ in
  $\mathbb{F}_2[X]/(X^p-1)$ with $\wt(f)\leq \delta p$. Then we have
$$
|S_{\delta}|=2^{pH'(\delta)+o(p)}
$$
where $H'(\delta)=H(\delta)/\log(2)$ and
$$
H(\delta)=-\delta\log(\delta)-(1-\delta)\log(1-\delta)
$$
is the entropy for Bernoulli random variables.
\end{lem}

\begin{proof}[Sketch of proof]
We have
$$
\binom{p}{\lfloor \delta p\rfloor}\leq |S_{\delta}|\leq
\sum_{j=1}^{\delta p}\binom{p}{j}\leq p\binom{p}{\lfloor \delta
  p\rfloor}
$$
which the Stirling formula reveals to be of size
$$
e^{H(\delta)p+o(p)}=2^{pH(\delta)/\log(2)+o(p)},
$$
as claimed.
\end{proof}

We also recall some fairly classical results on primes where $2$ has
relatively small multiplicative order.

\begin{lem}\label{lm-split}
\emph{(1)} For any $\epsilon$ with $0<\epsilon<1$, there exist infinitely
many primes $p$ such that $\ord_p(2)<\epsilon \cdot p$.
\par
\emph{(2)} Assume the Generalized Riemann Hypothesis for Dedekind zeta
functions of number fields. For any $\epsilon>0$, there exist
infinitely many primes $p$ such that $\ord_p(2)<p^{3/4+\epsilon}$.
\end{lem}

\begin{proof}
  In both cases, we use the criterion of Lemma~\ref{lm-chebo}: if
  $\ell$ is an odd prime and if $p$ is an odd prime distinct from
  $\ell$ such that $p$ is totally split in the field
  $K_{\ell}=\mathbb{Q}(e^{2i\pi/\ell},\sqrt[\ell]{2})$, then
  $p\equiv 1 \pmod\ell$ and the order of $2$ modulo $p$ divides
  $(p-1)/\ell$, hence is $<p/\ell$.
\par
Hence, taking $\ell$ to be any prime such that $\ell>1/\epsilon$, the
first statement follows from the existence of infinitely many primes
totally split in $K_{\ell}$ (this is an easy consequence of the
Chebotarev Density Theorem, see for
instance~\cite[Th. 13.4]{neukirch}).
\par
For the second, we use the explicit form of the Chebotarev Density
Theorem, following Serre's presentation of the results of Lagarias and
Odlyzko: for any odd prime $\ell$ and any $X\geq 2$, the number
$\pi_{\ell}(X)$ of primes $\leq X$ which are totally split in
$K_{\ell}$ satisfies
$$
\pi_{\ell}(X)=\frac{1}{[K_{\ell}:\mathbb{Q}]} \int_2^X\frac{dt}{\log
  t}+ O(\sqrt{X}\log(\ell X))
$$
where the implied constant is absolute, under the assumption that
Dedekind zeta functions satisfy the Riemann Hypothesis. Precisely,
this follows from~\cite[Th. 4]{serre}, applied with $E=K_{\ell}$,
$K=\mathbb{Q}$ and $C$ the trivial conjugacy class of the identity
element; then $n_E=[K_{\ell}:\mathbb{Q}]$ and the discriminant $d_E$
is estimated using the bound~\cite[(20)]{serre}.
\par
In particular, since the integral is of size $X/(\log X)$ and
$[K_{\ell}:\mathbb{Q}]\leq \ell^2$, this result shows that if
$\epsilon>0$ is fixed and $\ell$ is any prime large enough, there
exists a prime $p$ totally split in $K_{\ell}$ with
$p\leq \ell^{4+\epsilon}$. Such a prime $p$ satisfies
$$
\ord_p(2)<\frac{p}{\ell}<p^{1-1/(4+\epsilon)},
$$
and the result follows.
\end{proof}

The interest of these statements is that if the order $r$ of $2$
modulo $p$ is ``small'' compared with $p$, then by the discussion
following Proposition~\ref{decom}, the ring
$R=\mathbb{F}_2[X]/(X^p-1)$ contains many ideals. In particular, if
$r=p^{3/4+\epsilon}$ and $\eta$ with $0<\eta<1$ is fixed, and if we
look for ideals of dimension $ir\approx \eta p$, then for such primes
we have approximately $\binom{s}{i}$ ideals of dimension $\eta p$,
where (see Proposition~\ref{decom}), we have $s=(p-1)/r$ and
$i=\eta p/r\sim \eta s$. By Stirling's formula, as in the
Lemma~\ref{lm-entropy}, this numbers grows exponentially with $s$.


\subsection{Picking ideals at random}

Fix some real number with $0<\eta<1$. Let $p$ be a prime such that
there exists an ideal $I$ in $R=\mathbb{F}_2[X]/(X^p-1)$ with
$\dim(I)\sim \eta p$. 

Let $\delta>0$ be another parameter. Assuming that the probability for
an element of $I_p$ to be in the set $S_{\delta}$ of
Lemma~\ref{lm-entropy} is approximately the same as the probability
for a general element of $R$, the expected cardinality of the
intersection $S_{\delta}\cap I$ should be about
$$
2^{pH'(\delta)+\dim(I)-p+o(1)}=2^{p(H'(\delta)-(1-\eta))+o(1)}
$$
by Lemma~\ref{lm-entropy}. If $\eta$ and $\delta$ are chosen so that
$$
1-\eta>H'(\delta),
$$
this expectation is $<1$. So, as in the Borel-Cantelli lemma, if we
select an ideal $I_p$ of this approximate dimension for all primes
where this is possible (an infinite set, by Lemma~\ref{lm-split} and
Proposition~\ref{decom}), we may expect that only finitely many $p$
will have the property that $I_p$ intersects $S_{\delta}$. Since
$H'(\delta)\to 0$ as $\delta\to 0$, a suitable choice of $\delta$
exists for any fixed $\eta$. 
\par
Moreover, under the Generalized Riemann Hypothesis, picking the primes
$p$ as given by Lemma~\ref{lm-split} (2), the number of options for
$I_p$ grows exponentially as a function of
$s=p/\ord_p(2)\approx p^{1/4-\epsilon}$, and we need to succeed only
with a single one of them to obtain a good cyclic code with rate
$\eta$.

\subsection{The weak uncertainty principle should hold}

Here we give a heuristic argument, suggested by B. Poonen, as to why
the weak uncertainty principle of Definition~\ref{def-wup-2} should
hold for the field $\mathbb{F}_2$ for an infinite sequence of
primes. This is a variant of the previous argument.


First, the Generalized Riemann Hypothesis implies that there are
infinitely many primes such that $\ord_p(2) = \frac{p-1}{2}$ (this is
a simple variant of the argument of Hooley~\cite{hooley} for primitive
roots, where we count primes that are split in the quadratic field
$\mathbb{Q}(\sqrt{2})$, and not split in any field
$\mathbb{Q}(e^{2i\pi/\ell},\sqrt[\ell]{2})$ for $\ell\geq 3$ prime,
see Lemma~\ref{lm-chebo} and \cite{moree}).

We consider such primes and explain that all but finitely many should
satisfy Definition~\ref{def-wup-2} with $\epsilon=1/2$ and
$\delta=3/5$. Indeed, the condition $\ord_p(2)<\epsilon p$ holds by
construction.  Suppose $\mu_{F,p}\leq \delta p$. Then there exists a
non-zero $f\in \mathbb{F}_2[X]$ of degree $<p$ such that
\begin{equation}\label{eq-bad}
\mu_{F,p}(f)=\wt(f)+\dim I_f=\wt(f)+p-\deg(\gcd(f,X^p-1))\leq \delta
p.
\end{equation}
Since $\ord_p(2)=(p-1)/2$, the polynomial $(X^p-1)/(X-1)$ has exactly
two irreducible factors of degree $(p-1)/2$. So the gcd of $f$ and
$X^p-1$ is of degree $1$, $(p-1)/2$ or $p-1$. In the first case, the
inequality~(\ref{eq-bad}) is clearly false. In the third case, we have
$f=(X^p-1)/(X-1)$, with $\wt(f)=p$, and again~(\ref{eq-bad}) is false.
So $f$ must be divisible by exactly one of the two factors of degree
$(p-1)/2$, say $f_1$, and then we must have $\wt(f)\leq p/10+1/2$
for~(\ref{eq-bad}) to hold.
\par
Now comes the heuristic argument, where we will assume that the
property of being divisible by $f_1$ and of having support of size
$\leq p/10$ are ``independent'': the number of polynomials $f$ of
degree $<p$ divisible by $f_1$ is about $2^{p/2}$, and on the other
hand, the number of polynomials $f$ of degree $<p$ with $\wt(f)<p/10$
is $2^{pH'(1/10)+o(p)}$ by Lemma~\ref{lm-entropy}.
Since 
$$
H'(1/10)=\frac{H(1/10)}{\log(2)} \simeq 0.47<1/2,
$$
we may hope that the expected number of polynomials in the
intersection is
$$
O(2^{(0.47-1/2)p})=O(2^{-3p/100})
$$
and since the sum of the series $\sum 2^{-3p/100}$ is finite, this
suggests (by analogy with the Borel-Cantelli lemma) that the set of
primes where the intersection is non-empty is finite.

 


F. Voloch has pointed out that one must be careful with this
heuristic. Indeed, let $C_p$, for $p$ odd, be the \emph{quadratic
  residue} code of dimension $(p-1)/2$, namely the cyclic code
corresponding to the principal ideal generated by the polynomial
$$
\prod_{a\in(\mathbb{F}_p^{\times})^2}(X-a)\in \mathbb[X].
$$
If the last step is taken literally, the previous argument suggests
that the family of the cyclic codes $C_p$, parameterized by primes $p$
such that $\mathrm{ord}_p(2)=(p-1)/2$, is good. However, assuming GRH,
Voloch's results~\cite{V} imply that this is not the case. 
\par
More precisely, Voloch shows, under the Generalized Riemann
Hypothesis, that there exist an infinite sequence of primes $p$ for
which the distance of the code $C_p$ is $\ll p(\log p)^{-1}$ (he
obtains an unconditonal bound of size $\ll p(\log\log p)^{-1}$).
Although the primes that he constructs in~\cite{V} do not necessarily
satisfy the condition $\mathrm{ord}_p(2)=(p-1)/2$ that we wish to
impose, we will now show that the two can be combined (as was
suggested to us by Voloch).
\par
Indeed, Voloch defines a sequence of Galois extensions
$L_{\ell}/\mathbb{Q}$ of degree about $(\ell-1)2^{\ell}$, for $\ell$ a
prime. He shows that if $p$ is totally split in $L_{\ell}$, then the
distance of $C_p$ is $\leq (p-1)/(2\ell)$ (for this purpose, he uses a
formula of Helleseth). It turns out that the splitting restrictions in
$L_{\ell}$ are compatible with those involved in constructing primes
with $\mathrm{ord}_p(2)=(p-1)/2$.  Under the Generalized Riemann
Hypothesis, one gets by following Hooley's method (see, e.g.,~\cite[\S
5]{moree}) that for a given odd prime $\ell$ and for $X\geq 2$, there
are roughly
$$
\frac{1}{[L_{\ell}:\mathbb{Q}]} \frac{X}{\log X}
+O\Bigl(\frac{X(\log\log X)}{(\log X)^2}\Bigr)
$$
primes $p\leq X$ satisfying all the desired combined splitting
conditions. Since the degree of $L_{\ell}$ over $\mathbb{Q}$ is about
$\ell 2^{\ell}$, we can find a prime $p$ of size about
$\exp(\exp(\ell))$ that satisfies the desired conditions. This
provides an infinite family of codes $C_p$ with distance
$\ll p/(\log\log p)$, under the Generalized Riemann Hypothesis.

Although this discussion shows that the heuristic argument cannot be
literally correct, the optimist might still hope that the events which
we consider are sufficiently independent to still lead to infinitely
many primes where the weak uncertainty principle holds. It is maybe a
positive sign that the primes given by Voloch's argument are rather
sparse, and even then, only a very slow decay of their distance is
proved.

\section*{Appendix}

\subsection*{Chebotarev's Theorem}

A well-known (but not the best-known!) result of Chebotarev~\cite{C}
states the following:

\begin{thm}[Chebotarev] \label{chebo} Let $p$ be a prime and
  $\xi = e^{\frac{2\pi i}{p}} \in \mathbb{C}$. Let $V$ be the
  Vandermonde matrix
  $V = (\xi^{ij})_{i,j=0}^{p-1} \in M_p(\mathbb{C})$. Then each minor
  of the matrix $V$ is invertible, i.e., we have
  $\det(V|_{A\times B}) \ne 0$ for any $A,B \subset \{0,\ldots,p-1\}$,
  $|A|=|B|$, where $V|_{A\times B}$ denotes the minor of $V$ with rows
  in $A$ and columns in $B$.
\end{thm}

Let $R=\mathbb{C}[X]/(X^p-1)$.  Then $R$ is a vector space over
$\mathbb{C}$ with basis the images of the monomials $e_i=X^i$ for
$0\leq i\leq p-1$. 

(A multiple of) the Fourier transform on $\mathbb{Z}/p\mathbb{Z}$ can
be interpreted as the linear map
$\mathcal{F}\colon f\mapsto \widehat{f}$ from $R$ to $R$ such that
$$
\widehat{f} = \sum_{i=0}^{p-1} f(\xi^{-i})X^i \in R.
$$
It is elementary that the matrix representing this linear map is
$V' =(\xi^{-ij})_{i,j=0}^{p-1} \in M_p(\mathbb{C})$. Then each minor
of the matrix $V$ has a non-zero determinant if and only if the same
property holds for the matrix $V'$, so we may replace $V$ by $V'$ in
proving Chebotarev's Theorem.

We now show that Theorem~\ref{chebo} is \emph{equivalent} to the
uncertainty principle over $\mathbb{C}$. For a direct simple proof of
Chebotarev's Theorem, see the note~\cite{F} of Frenkel.

\begin{pro}
  Chebotarev's Theorem~\ref{chebo} is equivalent to the uncertainty
  principle for $\mathbb{Z}/p\mathbb{Z}$ over $\mathbb{C}$, i.e., to
  Theorem~\ref{up-abelian}.
\end{pro}

\begin{proof}
  For each $A \subset \{0,\ldots,p-1 \}$, we denote by $\ell^2(A)$ the
  space of elements of $R$ which have zero coefficients for the basis
  vectors $e_i$ for $i\notin A$, i.e., polynomials $f$ with support
  contained in $A$. For an element
$$
f=\sum_i a_iX^i\in R
$$
we denote by $f|_A$ the element
$$
\sum_{i\in A}a_iX^i
$$
of $\ell^2(A)$.

For any two subsets $A$ and $B$ of $\{0,\ldots,p-1\}$ with the same
cardinality, the linear map $T_{A,B}\colon \ell^2(A)\to \ell^2(B)$
obtained by restricting the Fourier transform (i.e.,
$T_{A,B}(f)=\widehat{f}|_{B}$ for $f\in\ell^2(A)$) is represented by
the matrix $V'_{A\times B}$ with respect to the bases $(e_i)_{i\in A}$
and $(e_i)_{i\in B}$.

  (Theorem~\ref{chebo} $\Rightarrow$ Theorem~\ref{up-abelian}) Assume
  for contradiction that there exists a non-zero element
$$
f=\sum_{i=0}^{p-1}a_iX^i \in \mathbb{C}[X]
$$ 
such that $|\supp(f)| + |\supp(\widehat{f})| \leq p$. Let
$A=\supp(f)$. Since $|\supp(\widehat{f})|\leq p-|A|$, the complement
of $\supp(\widehat{f})$ has cardinality $\geq |A|$. We can therefore
find a subset $B$ of the complement of $\supp(\widehat{f})$ such that
$|B|=|A|$.  Let $T=T_{A,B}:\ell^2(A) \rightarrow \ell^2(B)$. We then
have $T(f) = \widehat{f}|_B = 0$ since $B$ is in the complement of the
support of $\widehat{f}$, but $f$ is non-zero in $\ell^2(A)$. Hence
$T$ is not invertible.  Hence, by the previous remark, the matrix
$V'_{A\times B}$ has determinant zero, which contradicts Chebotarev's
Theorem.

(Theorem~\ref{chebo} $\Leftarrow$ Theorem~\ref{up-abelian}) Now assume
that there exist subsets $A,B \subset \{0,\ldots,p-1\}$ with $|A|=|B|$
and $\det(V'|_{A \times B})=0$.  This means that the linear map
$T=T_{A,B}: l^2(A) \rightarrow l^2(B)$ is not invertible. In
particular, $T$ is not injective. Let $f\not=0$ be an element of
$\ell^2(A)$ such that $0=T(f)=\widehat{f}|_B$.  Then
$\supp(f) \subset A$ and $B$ is contained in the complement of the
support of $\widehat{f}$. Hence
$$
|\supp(f)| \leq |A| = |B| \leq p - |\supp(\widehat{f})|,
$$
which contradicts the uncertainty principle.
\end{proof}

In this argument, we may replace $\mathbb{C}$ with any other field $F$
containing a $p$-primitive root of unity $\xi$.  So for any prime $p$
and for any field $F$ containing a $p$-primitive root of unity $\xi$,
Theorem~\ref{chebo} with respect to the prime $p$ (i.e. the claim that
each minor of the $p \times p$ Vandermonde matrix $(\xi^{ij})_{i,j}$
is invertible) is equivalent to the uncertainty principle for the
field $F$ with respect to $p$, i.e., to the claim that
$\mu_{F,p} > p$.



\bigskip
\textsc{Shai Evra}, Einstein Institute of Mathematics, The Hebrew University of
Jerusalem, 91904, Jerusalem, Israel.
\par \textit{E-mail address:} \texttt{shai.evra@gmail.com}

\bigskip
\textsc{Emmanuel Kowalski}, ETH Z\"{u}rich - D-MATH, R\"{a}mistrasse 101, CH-8092 Z\"{u}rich, Switzerland.
\par \textit{E-mail address:} \texttt{kowalski@math.ethz.ch}

\bigskip
\textsc{Alexander Lubotzky}, Einstein Institute of Mathematics, The Hebrew University of
Jerusalem, 91904, Jerusalem, Israel. 
\par \textit{E-mail address:} \texttt{alex.lubotzky@mail.huji.ac.il}

\end{document}